\documentclass[conference,letterpaper]{IEEEtran}

\usepackage{bm}

\usepackage{relsize}
\usepackage[utf8]{inputenc}
\usepackage[english]{babel}
\usepackage[T1]{fontenc}
\usepackage{hyperref}
\usepackage{graphicx}
\usepackage{multicol}

\usepackage{amsthm}
\usepackage{amssymb}
 
\newtheorem{theorem}{Theorem}
\newtheorem{corollary}[theorem]{Corollary}

\newtheorem{definition}[theorem]{Definition}

\usepackage[utf8]{inputenc} 
\usepackage[T1]{fontenc}
\usepackage{url}
\usepackage{ifthen}
\usepackage{cite}
\usepackage[cmex10]{amsmath} 

\begin{document}
\title{Changing the Local-Dimension of an Entanglement-Assisted Stabilizer Code Removes Entanglement Need} 


\author{%
  \IEEEauthorblockN{Lane G.~Gunderman}
  \IEEEauthorblockA{The Institute for Quantum Computing and
  Department of Physics and Astronomy,\\
  University of Waterloo, Waterloo, Ontario, N2L 3G1, Canada\\
                    Email: lgunderman@uwaterloo.ca}
}


\maketitle

\begin{abstract}
Having protected quantum information is essential to perform quantum computations. One possibility is to reduce the number of particles needing to be protected from noise and instead use systems with more states, so called qudit quantum computers. In this paper we show that codes for these systems can be derived from already known codes, and in particular this procedure removes the need for prior shared entanglement in entanglement-assisted quantum error-correcting codes, which is a result which could prove to be useful for fault-tolerant qudit, and even qubit, quantum computers as well as certain quantum communication tasks.
\end{abstract}

\section{Introduction}

Having protected quantum information is an essential piece of being able to perform controlled quantum computation operations. There are a variety of methods to help protect quantum information such as those discussed in \cite{lidar2013quantum}. In this paper we focus on entanglement-assisted quantum error-correcting codes (EAQECC) as hinted at in passing in \cite{bennett1996mixed}, but fully developed in generality in \cite{brun2006correcting}. EAQECC are similar in nature to stabilizer codes--the quantum analog of linear codes--but differ in their use of entanglement to allow for the immediate importing of many more classical linear codes to the quantum setting. Entanglement is a central resource in quantum computing and the sensitivity of entanglement to errors means that we ought to protect these entangled particles as well, meaning that higher-order error-correction will be required, however, EAQECC schematically assumes that single particles from these entangled pairs are already safely with the receiving party. The qudit version of EAQECC was shown in \cite{luo2017non}. In this work we consider EAQECC and show that it is possible to remove this entanglement requirement upon changing the local-dimension (or number of levels for each particle) that the code is applied to compared to the code's originally designed local-dimension. We further show that not only can we remove the entanglement need, but also we can at least preserve the distance of the code, and so improve the utility of these code, so long as entanglement is not completely free.

An EAQECC is specified by a set of partly non-commuting Pauli generators. Those generators in the code that do not commute can be written such that any non-commutation relations are resolved through the use of an entangled pair of particles and the superdense coding protocol \cite{lidar2013quantum}. In the qudit case the entangled pair of particles used is given by:
\begin{equation}
    |\Phi_+^q\rangle=\frac{1}{\sqrt{q}}\sum_{i=0}^{q-1} |i,i\rangle.
\end{equation}
EAQECC allow the dual code space constraint from the Calderbank-Shor-Steane (CSS) theorem to be ignored--allowing any classical code to be imported to the quantum case \cite{brun2006correcting}. This allows for immediate translation of classical error-correcting codes into quantum error-correcting codes so long as a source of shared entanglement is available.

Even with error-correcting codes having sufficient amounts of protected quantum information to perform useful tasks is still an unresolved challenge. A way to retain a similarly sized computational space while reducing the number of particles that need precise controls and carefully regulated environments is to replace the standard choice of qubits with \textit{qudits}, quantum particles with $q$ levels. Throughout this work we require $q$ to be a prime so that each nonzero element has a unique multiplicative inverse over $\mathbb{Z}_q$. This restriction can likely be removed, but for simplicity and clarity we only consider this case. Experimental realizations of these systems are currently underway \cite{qudition,quditlight,molecules}, so having more error-correcting codes will aid in protecting such systems. Prior work on qudit error-correcting codes have often had challenging restrictions between the parameters of the code \cite{quditgeneral,quditbch,quditmds}, and we've already made progress on reducing this barrier in a prior paper \cite{gunderman2020local}. Our prior work showed the ability to make error-correcting codes that preserved their parameters, generally, even upon changing the local-dimension of the system. Beyond this, these systems also have proven connections to foundational aspects of physics \cite{contextuality}. Seeing these potential reasons for using qudits, this work builds off of our prior work.

\section{Definitions}
In this section, we recall common definitions and results for qudit operators. A qubit is defined as a two level system with states $|0\rangle$ and $|1\rangle$. We define a qudit as being a quantum system over $q$ levels, where $q$ is prime. Throughout we take $\mathbb{Z}_q$ as the set $\{0,1,\ldots,q-1\}$.

\begin{definition}\label{dim}
Generalized Paulis for a space over $q$ orthogonal levels are given by:
\begin{equation}
 X_q|j\rangle=|(j+1)\mod q\rangle,\quad Z_q|j\rangle=\omega^j|j\rangle
\end{equation}
with $\omega=e^{2\pi i/q}$, where $j\in\mathbb{Z}_q$. These Paulis form a group, denoted $\mathbb{P}_q$.
\end{definition}

When $q=2$, these are the standard qubit operators $X$ and $Z$, with $Y=iXZ$. This group structure is preserved over tensor products since each of these Paulis has order $q$.

As shown in \cite{brun2006correcting} and \cite{luo2017non}, an EAQECC is specified by $s$ commuting Pauli operators and a set of $c$ Pauli operator pairs $\{\mathcal{X}_i,\mathcal{Y}_i\}$ that do not commute:
\begin{equation}
\mathcal{X}_i\odot \mathcal{Y}_i\neq 0, \quad \forall i    
\end{equation}
while all other operators commute. We let $k=s+2c$ be the total number of $n$-qudit generalized Pauli operator generators used to specify the code. Note that this is slightly different from the standard choice as this work focuses on the total number of generators opposed to many works which focus more on the number of encoded particles.

Although the generators in an EAQECC do not all commute, they do form a group. The entirety of the group, with the scalar coefficient quotiented out, is composed of all possible compositions ($\circ$) of the generators. This forms a subgroup of size $q^k$ as each generator has order $q$. This then leads to there being $q^{n+c-k}$ orthonormal basis states, or codewords, where the additive factor of $c$ in the exponent is due to the added space of the entangled particles.

Finding the commutator of these generators with an error provides the \textit{syndrome} of that error. These syndromes provide insight into which error may have occurred so that we can determine the error and potentially undo it. The standard choice of error model is the depolarizing channel which depends on the weights of the errors:

\begin{definition}
The weight of an $n$-qudit Pauli operator is given by the number of non-identity operators in it.
\end{definition}

\begin{definition}
An EAQECC, specified by its generators, is characterized by a set of parameters:
\begin{itemize}
\item $n$: the number of particles that are transmitted through the code (in the traditional communication setting this is the number of particles sent to the receiver)
\item $n+c-k$: the number of encoded (logical) qudits
\item $d$ (for non-degenerate codes (where all group members have weight at least $d$)): the distance of the code, given by the lowest weight of an undetectable generalized Pauli error (commutes with all elements of the group, but is not in the group itself)
\item $c$: the number of entangled pairs needed in order to resolve commutation relations between the generators of the code
\end{itemize}
These values are specified for a particular code as: $[[n, n+c-k,d;c]]_q$, where $q$ is the local-dimension of the qudits.
\end{definition}

The minimal value of $c$ needed for a particular set of generators was shown in \cite{wilde2008optimal}. Working with tensors of operators can be challenging, and so we make use of the following well-known mapping from these to vectors, following the notation from \cite{gunderman2020local}. This representation is often times called the symplectic representation for the operators, but we use this notation instead to allow for greater flexibility. This linear algebraic representation will be used for our proofs.

\begin{definition}[$\phi$ representation of a qudit operator]
We define the surjective map: 
\begin{equation}
\phi_q: \mathbb{P}_q^n\mapsto \mathbb{Z}_q^{2n}
\end{equation}
which carries an $n$-qudit Pauli in $\mathbb{P}_q^n$ to a $2n$ vector mod $q$, where we define this map as:
\begin{multline}
\phi_q(\omega^\alpha \otimes_{i-1} I\otimes X_q^a Z_q^b\otimes_{n-i} I)\\
=( 0^{ i-1}\ a\ 0^{ n-i} | 0^{ i-1}\ b\ 0^{ n-i}),
\end{multline}
which puts the power of the $i$-th $X$ operator in the $i$-th position and the power of the $i$-th $Z$ operator in the $(i+n)$-th position of the output vector. This mapping is defined as a homomorphism with: $\phi_q(s_1\circ s_2)=\phi_q(s_1)\oplus \phi_q(s_2)$, where $\oplus$ is component-wise addition mod $q$. We denote the first half of the vector as $\phi_{q,x}$ and the second half as $\phi_{q,z}$.
\end{definition}

We may invert the map $\phi_q$ to return to the original $n$-qudit Pauli operator with the global phase being undetermined. We make note of a special case of the $\phi$ representation:

\begin{definition}
Let $q$ be the dimension of the initial system. Then we denote by $\phi_\infty$ the mapping:
\begin{equation}
    \phi_\infty:  \mathbb{P}_q^n\mapsto \mathbb{Z}^{2n}
\end{equation}
where no longer are any operations taken $\mod$ some base, but instead carried over the full set of integers.
\end{definition}

The ability to generally define $\phi_\infty$ as a homomorphism still (and with the same rule) is a portion of the results of this paper--shown in Theorem \ref{inv}. $\phi_q$ is the standard choice for working over $q$ bases, however, our $\phi_\infty$ allows us to avoid being dependent on the local-dimension of our system when working with our code. Formally we will write a code in $\phi_q$, perform some operations, then write it in $\phi_\infty$, then select a new local-dimension $q'$ and use $\phi_{q'}$. We shorten this to write it as $\phi_\infty$, and can later select to write it as $\phi_{q'}$ for some prime $q'$ by taking element-wise $\mod q'$.

The commutator of two operators in this picture is given by the following definition:
\begin{definition}
Let $s_i,s_j$ be two qudit Pauli operators over $q$ bases, then these commute if and only if:
\begin{equation}
\phi_q(s_i)\odot \phi_q(s_j)=0\mod q
\end{equation}
where $\odot$ is the symplectic product, defined by:
\begin{multline}
\phi_q(s_i)\odot \phi_q(s_j)\\ =\oplus_k [\phi_{q,z}(s_j)_k\cdot  \phi_{q,x}(s_i)_k- \phi_{q,x}(s_j)_k \cdot \phi_{q,z}(s_i)_k]
\end{multline}
where $\cdot$ is standard integer multiplication $\mod q$ and $\oplus$ is addition $\mod q$.
\end{definition}

When the commutator of $s_i$ and $s_j$ is not zero, this provides the difference in the number of $X$ operators in $s_i$ that must pass a $Z$ operator in $s_j$ and the number of $Z$ operators in $s_i$ that must pass an $X$ operator in $s_j$ when attempting to switch the order of these two operators. We will use these values, without taking modulo $q$, to prove Theorem \ref{inv}.

Before finishing, we make a brief list of some possible operations we can perform on our $\phi$ representation for an EAQECC:
\begin{enumerate}
    \item We may perform elementary row operations over $\mathbb{Z}_q$, corresponding to relabelling and composing generators together.
    \item We may swap registers (qudits) in the following ways:
        \begin{enumerate}
            \item We may swap columns $(i,i+n)$ and $(j,j+n)$ for $0<i,j\leq n$, corresponding to relabelling qudits.
            \item We may swap columns $i$ and $(-1)\cdot (i+n)$, for $0<i\leq n$, corresponding to conjugating by a Hadamard gate on particle $i$ (or Discrete Fourier Transforms in the qudit case \cite{qudit}) thus swapping $X$ and $Z$'s roles on that qudit.
        \end{enumerate}
\end{enumerate}

All of these operations leave the code parameters $n$, $k$, and $d$ alone, but can be used in proofs.

\section{Local-Dimension-Invariant EAQECC}

In this section we prove how to remove the entanglement need from EAQECC and under what conditions we can promise the distance of the code is at least preserved. To this end we begin with a loosened definition of local-dimension-invariant stabilizer codes which were first introduced in \cite{gunderman2020local}:

\begin{definition}
A code is called \textit{effectively local-dimension-invariant} if all generators commute over $p$ levels, $p\neq q$, upon evaluating all entries at a pre-determined function of $p$ while the original code over $q$ is unchanged.
\end{definition}

If we can transformed an EAQECC into an effectively local-dimension-invariant code, the non-commuting generators are transformed into commuting generators and so removing the need for entanglement. Being able to transform any particular EAQECC into an effectively local-dimension-invariant code is one challenge, however, as we will show shortly not only can \textit{all} EAQECC codes be turned into an effectively local-dimension-invariant form, but we also provide a prescriptive technique to transform any given code.

The key observation needed to show this is that we may break up the commutator of the generators over the integers into two parts each removed through a different technique. Let $c_{ij}=\phi_\infty(s_i)\odot\phi_\infty(s_j)$ for generators $s_i,s_j$ in the code. Define $\alpha_{ij}=c_{ij}\mod q$ so that $c_{ij}=\alpha_{ij}+m_{ij}q$ for some integer $m_{ij}$. We note that we can always rewrite this as $c_{ij}=\alpha_{ij}+n_{ij}q+(m_{ij}-n_{ij})q$ such that $\alpha_{ij}+n_{ij}q\mod p=0$, which will allow for the removal of all entanglement requirements from the code, once we remove $(m_{ij}-n_{ij})q$. We will remove this remaining term through the addition of a lower triangular $k\times k$ matrix, $L$, where the exact values of the entries are a pre-determined function of the new local-dimension $p$. The following proof is similar to the original invariant procedure from \cite{gunderman2020local}, but requires care around particular cases.

\begin{theorem}\label{inv}
All EAQECC over $q$ levels can be made into an \textit{effectively local-dimension-invariant} stabilizer code.
\end{theorem}

\begin{proof}
Let $S$ be an EAQECC with parameters $[[n,n+c-k,d; c]]_q$, we may write this code as $\phi_q(S)$. When using the initial generators of $S$ in $\phi_q(S)$, the symplectic product matrix $[\odot]_q$, containing all pairwise commutator values taken mod $q$, will have exactly $2c$ nonzero entries corresponding to the generators whose commutators need to be resolved via entanglement. We will now allow the number of nonzero entries to change by transforming $\phi_q(S)$ via the rules outlined earlier to an EAQECC canonical form:
\begin{equation}
    \phi_q(S)=\begin{bmatrix}
     I_k & X_2 & | & Z_1 & Z_2
    \end{bmatrix}
\end{equation}
where $Z_1$ is a $k\times k$ matrix, and $X_2$ and $Z_2$ are $k\times (n-k)$ matrices. Let $[\odot]_\infty$ be the anti-symmetric symplectic commutator matrix, written over the integers. We will add a lower triangular matrix $L$ to $Z_1$ such that after this addition we leave the code alone over $\mod q$, and yet have $[\odot']_p=\mathbf{0}$ upon evaluation for any choice of $p$, with $p\neq q$.

Let $c_{ij}=[[\odot]_\infty]_{ij}=\phi_\infty(s_i)\odot \phi_\infty(s_j)$. Upon addition of the $L$ matrix, our updated generators $S'$ will be given by:
\begin{equation}
    \phi_\infty(S')=\begin{bmatrix}
     I_k & X_2 & | & Z_1+L\ & Z_2
    \end{bmatrix}.
\end{equation}

\begin{figure*}[!htbp]
\normalsize
\begin{eqnarray}\label{steps}
    \phi_\infty (s_i')\odot \phi_\infty (s_j')&=& [\phi_\infty(s_i)+(0\ |\ L_i\ 0)]\odot [\phi_\infty(s_j)+(0\ |\ L_j\ 0)]\\
    &=& \phi_\infty(s_i)\odot \phi_\infty (s_j)+\phi_\infty(s_i)\odot (0\ |\ L_j\ 0)+(0\ |\ L_i\ 0)\odot \phi_\infty (s_j)+(0\ |\ L_i\ 0)\odot (0\ |\ L_j\ 0)\\
    &=& c_{ij}+0-L_{ij}+0\\
    &=& \alpha_{ij}+n_{ij}q+(m_{ij}-n_{ij})q-L_{ij}.\label{end}
\end{eqnarray}
\hrulefill
\vspace*{4pt}
\end{figure*}

For these updated generators, $S'$, we have from equations (\ref{steps})-(\ref{end}), when $i>j$:
\begin{equation}
    \phi_\infty (s_i')\odot \phi_\infty (s_j')=\alpha_{ij}+n_{ij}q+(m_{ij}-n_{ij})q-L_{ij}.
\end{equation}
When $i<j$ the commutator value is just the additive inverse of the $j>i$ case as the symplectic product matrix is anti-symmetric.

\if{false}
\begin{widetext}
\begin{eqnarray}
    \phi_\infty (s_i')\odot \phi_\infty (s_j')&=& [\phi_\infty(s_i)+(0\ |\ L_i\ 0)]\odot [\phi_\infty(s_j)+(0\ |\ L_j\ 0)]\\
    &=& \phi_\infty(s_i)\odot \phi_\infty (s_j)+\phi_\infty(s_i)\odot (0\ |\ L_j\ 0)+(0\ |\ L_i\ 0)\odot \phi_\infty (s_j)+(0\ |\ L_i\ 0)\odot (0\ |\ L_j\ 0)\\
    &=& c_{ij}+0-L_{ij}+0\\
    &=& \alpha_{ij}+n_{ij}q+(m_{ij}-n_{ij})q-L_{ij}.
\end{eqnarray}
\end{widetext}
\fi
Let $\nu=q\mod p$, with $\nu\in\mathbb{Z}_p$. Set $n_{ij}=-\nu^{-1}\alpha_{ij}$, then $\alpha_{ij}+n_{ij}\nu\mod p=0$. From this, we also have $\alpha_{ij}+n_{ij}\nu=\alpha_{ij}+n_{ij}q\mod p$, and so $\alpha_{ij}+n_{ij}q=0\mod p$, meaning that the first two terms in the updated commutator disappear upon evaluating at a chosen $p$ value. Lastly, setting $L_{ij}=(m_{ij}-n_{ij})q$ and adding this lower triangular matrix enforces commutation over $p$ by subtracting off the remaining term.
\end{proof}

Before moving on, we emphasize that for the above constructive proof, the only parameter that must be determined for a give value of $p$ is $\nu^{-1}$, and so they are only effectively local-dimension-invariant. A couple of remarks about $\nu^{-1}$ are in order. First, notice that in the above proof the inability to perform this invariant forming operation over $q$ is manifest as $\nu^{-1}$ is not defined as 0 never has a multiplicative inverse. Second, we make note a couple of crucial cases for $\nu$ in the above proof. When $p>q$ then $L_{ij}=c_{ij}$ and $\nu^{-1}=q^{-1}\mod p$. When $p=2$, $L_{ij}=c_{ij}+(q-1)\alpha_{ij}$. The only cases where the code is not truly local-dimension-invariant but only \textit{effectively} local-dimension-invariant is when $2<p<q$.

The above Theorem merely transforms an EAQECC into a set of commuting generators. Following Theorem 16 from \cite{gunderman2020local}, in order to make statements about the distance of this transformed code we must bound the maximal entry in $\phi_\infty(S')$:

\begin{corollary}
The maximal entry in $\phi_\infty(S')$, $B$, upon selecting a new local-dimension $p$, $p>q$, satisfies:
\begin{equation}
    B\leq [2+(n-k)(q-1)](q-1).
\end{equation}
\end{corollary}

\begin{proof}
We begin by noting from our prior work that we have $c_{ij}\leq B-(q-1)$, where $B$ is given by $[2+(n-k)(q-1)](q-1)$ \cite{gunderman2020local}. Here we have: \begin{eqnarray}
    L_{ij}&=&(m_{ij}+\nu^{-1}\alpha_{ij})q\\
    &=&\alpha_{ij}+m_{ij}q+(\nu^{-1}q-1)\alpha_{ij}\\
    &=&c_{ij}+(q^{-1}q-1)\alpha_{ij}\\
    &\leq & B-(q-1).
\end{eqnarray}
Then the maximal entry is upper bounded by $B-(q-1)+(q-1)=B$, which is the same as for stabilizer codes and any tightening on the bound of $B$ there will apply in this bound as well.
\end{proof}

We can now combine the ability to create effectively local-dimension-invariant forms for EAQECC codes and this entry bound to ensure that the EAQECC will have at least the same distance upon being transformed into a stabilizer code. According to Theorem 16 from \cite{gunderman2020local}, if brought from $q$ levels to $p$ levels with $p>p^*$, where $p^*=B^{2(d-1)}[2(d-1)]^{(d-1)}$ then the distance will be at least preserved. As the $n$ particles being transmitted can have their errors classified into the same subsets of undetectable errors as was used to prove the distance bound for stabilizer codes, this immediately provides the following theorem:

\begin{theorem}\label{primary}
We may transform any non-degenerate $[[n,n+c-k,d;c]]_q$ EAQECC into a $[[n,n-k,d';0]]_p$ stabilizer code with $d'\geq d$ so long as $p$ is a prime with $p>p^*$.
\end{theorem}

Beyond this, using the same reasoning as in our prior work, we can also define logical operators for these codes \cite{gunderman2020local}. Putting together all the results, we have defined quantum error-correcting codes which can protect information, remove entanglement use, and have logical operators, and while the distance of the code can only be promised at sufficiently many bases, it is possible to preserve the distance even below this cutoff as the following examples show.

We now apply this procedure to the $[[4,1,3;1]]_2$ code from \cite{brun2006correcting}, given by:
\begin{equation}
    \begin{matrix}
    X\\
    Z\\
    I\\
    I
    \end{matrix}\ \vline\ \begin{matrix}
     Z & X & Z & I\\
     Z & Z & I & Z\\
     Y & X & X & Z\\
     Z & Y & Y & X
    \end{matrix},
\end{equation}
where we have left the entanglement particle there for clarity. The four generators used to protect the transmitted particles can be written in the $\phi_2$ representation as:
\begin{equation}
    \begin{bmatrix}
     0 & 1 & 0 & 0 & | & 1 & 0 & 1 & 0\\
     0 & 0 & 0 & 0 & | & 1 & 1 & 0 & 1\\
     1 & 1 & 1 & 0 & | & 1 & 0 & 0 & 1\\
     0 & 1 & 1 & 1 & | & 1 & 1 & 1 & 0
    \end{bmatrix}.
\end{equation}
We put this into canonical form by applying a Hadamard on particle four, then performing RREF. Applying Theorem \ref{inv} we obtain an invariant form of:
\begin{equation}
    \begin{bmatrix}
     1 & 0 & 0 & 0 & | & 1 & 0 & 1 & 1\\
     0 & 1 & 0 & 0 & | & 1 & 0 & 1 & 0\\
     0 & 0 & 1 & 0 & | & 0 & 1 & 0 & 1\\
     0 & 0 & 0 & 1 & | & 1 & 2 & -1 & 0
    \end{bmatrix}
\end{equation}
This code has $d= 3$ for $p>3$ as no linear combination of columns corresponding to weight two Paulis are linearly dependent. We have transformed this into a $[[4,0,3;0]]_p$ code for $p>3$. Note that this does not mean for $p=3$ it is not possible to modify the code such that the distance is still preserved, just that this prescriptive method does not provide it given the canonical form used. We provide the $p=3$ case later, but wait to discuss it.

A more concise way to summarize this result is by considering the rate of this code upon performing this transformation. This technique alters the rate of an EAQECC in the following ways, following the definitions from \cite{wilde2010entanglement}:
\begin{itemize}
    \item The entanglement-assisted rate is altered from $(n+c-k)/n$ to $(n-k)/n$.
    \item The trade-off rate is altered from $((n+c-k)/n,c/n)$ to $((n-k)/n,0)$. 
    \item The catalytic rate is unchanged from $(n-k)/n$.
\end{itemize}
The correct choice of which rate definition to use depends on the application. The entanglement-assisted rate assumes that entanglement sharing is free, the trade-off rate allows for some unspecified cost for the entanglement, while the catalytic rate assumes that the entanglement costs roughly the same as transmitting a particle. So long as entanglement is not free these rate changes can be of use. These changes to the rates require the following pair of caveats that: one, the local-dimension must be changed, and, two, these rate changes are only proper so long as the distance of the code is also at least preserved. If the distance is not preserved, the rate will change still, but the quality of protection for the code has dropped making the comparison on unequal footing.

\section{Future Directions}

While the above example considered $p>q$, the following example shows that it is possible to have $p<q$ and still obtain at least as good parameters:
\begin{equation}\label{double}
    \begin{bmatrix}
     0 & 11 & 3 & 4 & | & 12 & 11 & 11 & 12\\
     14 & 6 & 14 & 9 & | & 13 & 8 & 5 & 0\\
     4 & 13 & 10 & 11 & | & 10 & 1 & 3 & 2\\
     0 & 13 & 4 & 9 & | & 11 & 5 & 0 & 0
    \end{bmatrix}.
\end{equation}
This is a $[[4,2,2;2]]_5$ code as well as a $[[4,0,3;0]]_3$ code. This provides the $p=3$ case of the example considered before, showing that we could achieve $p^*=q$. Not only does the application of this result remove the need for entanglement for this code, but it also improved the distance of the code.

Theorem \ref{primary} provides a promise on the distance of the code. For choices of $p<p^*$ one would need to computationally check whether the distance of the code is at least preserved. As remarked before, this procedure for making the code \textit{effectively local-dimension-invariant} is not unique. Even if the distance is not preserved at a value of $p$ using the given procedure, it does not mean that there is not another procedure which will preserve the distance of the code while obtaining the entanglement removal. We have not yet been able to find a procedure which always preserves the distance for $p>q$, but believe that it is possible, and so leave this as a future direction.

In this work we've proven a method to remove the needed entanglement in EAQECC upon changing the local-dimension as well as conditions to ensure the distance of the code remains. This result firstly allows for the removal of needing to send shared entangled particles between two parties in a communication setting, assuming they are using EAQECC. This means that only the $n$ particles must be sent and protected, removing the need to protect the shared entangled particles during the transmission of them. Secondly this result has implications in standard stabilizer error-correcting codes. Since EAQECC codes do not need to obey dual constraints from the CSS theorem, arbitrary classical codes can be imported into this setting, where they will require entanglement, but then this entanglement usage can be removed by altering the local-dimension and applying the method provided here. For instance, if $p^*$ can be reduced to $q$, a classical binary LDPC code could be imported using EAQECC on qubits, then using these methods the code could be used over qutrits \textit{without} requiring any entanglement--effectively producing an LDPC quantum code without needing entanglement. Unfortunately, in order to achieve this the value for $p^*$ must be decreased and the distance promise must be shown for the case of degenerate codes as well, as LDPC codes typically utilize the degeneracy to achieve their high rates with high distance.

The results shown here provide another use of local-dimension-invariant stabilizer codes, and so naturally there are questions as to what other uses this technique will have. In addition to this method, is it possible to apply this technique to show some foundational aspect of quantum measurements? Beyond this, this work also makes the challenge of reducing $p^*$ more crucial than before as it would reduce the number of values of $p$ that need to have their distance computationally checked.

\section*{Acknowledgments}

We thank Arun Moorthy for help with writing a program to generate the example from equation (\ref{double}) and David Cory, Mark Wilde, and Markus Grassl for helpful feedback.

\section*{Funding}
This work was supported by Industry Canada, the Canada
First Research Excellence Fund (CFREF), the Canadian Excellence Research Chairs (CERC 215284) program, the Natural Sciences and Engineering Research Council of Canada
(NSERC RGPIN-418579) Discovery program, the Canadian
Institute for Advanced Research (CIFAR), and the Province
of Ontario.

\bibliographystyle{unsrt}
\phantomsection  
\renewcommand*{\bibname}{References}

\bibliography{main}

\if{false}
\begin{abstract}
   Instructions
  are given for the preparation and submission of papers for the
  \emph{2021 International Symposium on Information Theory}.
\end{abstract}

\textit{A full version of this paper is accessible at:}
\url{https://arxiv.org/pdf/21xx.xxxx.pdf} 

\section{Submission of Papers for Review}

Papers in the form of a PDF file, formatted as described below, may be
submitted online at
\begin{center}
  \url{https://2021.ieee-isit.org/Papers.asp}
\end{center}
The deadline for registering the paper and uploading the manuscript is \textbf{January 27, 2021}.

A paper's primary content is restricted in length to \textbf{five pages} but
authors are allowed an optional sixth page only containing references.
The IEEEtran-conference style should be used as presented here.
Submissions should use a font size no smaller than 10 points and have reasonable margins on all the 4 sides of the text.

Each paper must be classified as ``eligible for student paper award''
or ``not eligible for student paper award''. \textbf{Do not include any text concerning the student paper award in the abstract or manuscript.}

\section{Submission of Accepted Papers}

Accepted papers will be published in full. A program booklet and a book of abstracts will also be
distributed at the Symposium to serve as a guide to the sessions.

The deadline for the submission of the final camera-ready paper will
be announced in due course.  Accepted papers not submitted by that
date will not appear in the ISIT proceedings and will not be included
in the technical program of the ISIT.

\section{Paper Format}

\subsection{Templates}

The paper (A4 or letter size, double-column format, not exceeding
5~pages) should be formatted as shown in this sample \LaTeX{} file
\cite{Laport:LaTeX, GMS:LaTeXComp, oetiker_latex, typesetmoser}.

The use of Microsoft Word instead of \LaTeX{} is strongly
discouraged. However, acceptable formatting may be achieved by using
the template that can be downloaded from the ISIT 2021 website:
\begin{center}
  \url{https://2021.ieee-isit.org/}
\end{center}

Users of other text processing systems should attempt to duplicate the
style of this example, in particular the sizes and type of font, as
closely as possible.

\subsection{Formatting}

The style of references, equations, figures, tables, etc., should be
the same as for the \emph{IEEE Transactions on Information
  Theory}. The source file of this template paper contains many more
instructions on how to format your paper. So, example code for
different numbers of authors, for figures and tables, and references
can be found (they are commented out).

For instructions on how to typeset math, in particular for equation
arrays with broken equations, we refer to \cite{typesetmoser}.

Final papers should have no page numbers and no headers or footers (both will be added during the production of the proceedings).
The top and bottom margins should be at least 0.5 inches to leave room for page numbers.
The affiliation shown for authors should constitute a sufficient mailing
address for persons who wish to write for more details about the
paper.
All fonts should be embedded in the pdf file.

\subsection{PDF Requirements}

Only electronic submissions in form of a PDF file will be
accepted. The PDF file has to be PDF/A compliant. A common problem is
missing fonts. Make sure that all fonts are embedded. (In some cases,
printing a PDF to a PostScript file, and then creating a new PDF with
Acrobat Distiller, may do the trick.) More information (including
suitable Acrobat Distiller Settings) is available from the IEEE
website \cite{IEEE:pdfsettings, IEEE:AuthorToolbox}.

\section{Conclusion}

We conclude by pointing out that on the last page the columns need to
balanced. Instructions for that purpose are given in the source file.

Moreover, example code for an appendix (or appendices) can also be
found in the source file (they are commented out).


\section*{Acknowledgment}

We are indebted to Michael Shell for maintaining and improving
\texttt{IEEEtran.cls}.

\IEEEtriggeratref{4}



\fi
\end{document}